\documentclass[a4paper,10pt, reqno]{amsart}
\usepackage[latin1]{inputenc}
\usepackage{amssymb}
\usepackage{amsmath}
\usepackage{amsthm}
\usepackage{graphicx}
\usepackage{a4wide}
\newtheorem{thm}{Theorem}
\newtheorem{lem}{Lemma}
\newtheorem{prop}{Proposition}
\newtheorem{rem}{Remark}
\newtheorem{cor}{Corollary}
\newtheorem{ass}{Assumption}
 \usepackage[colorlinks=true]{hyperref}
\hypersetup{urlcolor=blue, citecolor=red}
\newcommand{\beq}{\begin{eqnarray}}
\newcommand{\eeq}{\end{eqnarray}}
\newcommand{\beqs}{\begin{eqnarray*}}
\newcommand{\eeqs}{\end{eqnarray*}}
\newcommand{\bequ}{\begin{equation}}
\newcommand{\eequ}{\end{equation}}

\def\r{\rho}

\newcommand{\cal}{\mathcal}
\newcommand{\pa}{\partial}

\def\e{\varepsilon}

\let\pa\partial

\let\eps\varepsilon
\title{Kinetic derivation of fractional Stokes and Stokes-Fourier systems}

\author{S. Hittmeir}\address{Sabine Hittmeir: RICAM Linz, Austrian Academy of Sciences, Altenberger Str. 69, 4040 Linz, Austria, (sabine.hittmeir@oeaw.ac.at)} 
\author{S. Merino-Aceituno}\address{S. Merino-Aceituno: Cambridge Centre for Analysis, Centre for Mathematical Sciences, University of Cambridge, Wilberforce Road, Cambridge CB3 0WA, UK (s.merino-aceituno@maths.cam.ac.uk)}
\thanks{This work has been initiated during inspiring discussions with Clement Mouhot, to whom the authors want to thank for many helpful advices and encouragement. The authors are also grateful to James Norris for discussions that lead to a better understanding of the problem on the probabilistic side. 
\medskip\\
SH thanks to the financial support by the Austrian Academy of Sciences \"OAW via
the New Frontiers project NST-0001 and acknowledges the support of the King Abdullah University of Science and Technology (KAUST) within
grant KUK-I1-007-43 during the research stay in Cambridge, where this collaboration has been started.  \\
The work of SMA was supported by the UK Engineering and Physical Science Research Council (EPSRC) grant EP/H023348/1 for the University of Cambridge Centre for Doctoral Training, the Cambridge Centre for Analysis; and also by a Fellowship for Graduate Courses funded by the Fundacion Caja Madrid. Finally, SMA thanks to the Wolfang Pauli Institute in Vienna for its hospitality and support during her stays there.}
\begin{document}

\begin{abstract}
In recent works it has been demonstrated that using an appropriate rescaling, linear Boltzmann-type equations give rise to a scalar fractional diffusion equation in the limit of a small mean free path. The equilibrium distributions are typically heavy-tailed distributions, but also classical Gaussian equilibrium distributions allow for this phenomena if combined with a degenerate collision frequency for small velocities. This work aims to an extension in the sense that a linear BGK-type equation conserving not only mass, but also momentum and energy,  for both mentioned regimes of equilibrium distributions is considered. In the hydrodynamic limit we obtain a fractional diffusion equation for the temperature and density making use of the Boussinesq relation and we also demonstrate that with the same rescaling fractional diffusion cannot be derived additionally for the momentum. But considering the case of conservation of mass and momentum only, we do obtain the incompressible Stokes equation with fractional diffusion in the hydrodynamic limit for heavy-tailed equilibria.
\end{abstract}
\maketitle
\medskip\\


\pagestyle{myheadings}
\thispagestyle{plain}
\section{Introduction}
Asymptotic analysis for kinetic transport equations of Boltzmann-type is a very classical problem. For a collision operator conserving mass, it is well known that for a small mean-free path and under an appropriate rescaling a diffusion equation for the density of particles can be obtained, see e.g. \cite{BLP}, \cite{DGP}, \cite{BDD}. This diffusive approximation requires the equilibrium distributions, which are classically a Gaussian, to have finite variance. On the other hand  heavy-tailed equilibrium distributions violating this finite variance condition do appear in  many contexts, e.g. in astrophysical plasmas \cite{ST}, in mixtures of Maxwell gases \cite{BG} or in applications in economy \cite{DT}.  For such types of equilibrium distributions the diffusive time scale is too long. Mellet et al. \cite{MMM} showed that under an appropriate rescaling of the linear Boltzmann equation in this case a fractional diffusion equation can be derived in the macroscopic limit. Their proof relies on the use of the Fourier-Transform. 
Similar results to \cite{MMM} have also been derived from a probabilistic approach in \cite{JKO}. Mellet \cite{M} rederived these results using a new moments method that also allows for considering space-dependent collision operators, therefore providing a more suitable framework for addressing nonlinear problems. Using the same method it was shown in \cite{BMP} that the phenomenon of fractional diffusion can also arise out of a degenerate collision frequency for small velocities, where here the equilibrium distribution can be chosen as the classical Gaussian.  
In \cite{BMP2} the authors demonstrate how a Hilbert expansion approach can be successfully employed to derive fractional diffusion equations from linear kinetic transport equations. This method requires stronger assumptions on the initial data than the moments method, but therefore also provides better convergence results. 

In this work we aim to extend this fractional diffusion limit to a kinetic transport equation conserving not only mass, but also momentum and energy. Many works have been investigating the incompressible fluid dynamical limit of the Boltzmann equation, see e.g. \cite{BU}, \cite{GL}, \cite{S} and references therein. We shall review the basic formal derivation of the linear equations of the corresponding hydrodynamic limit below. On the other hand Navier-Stokes type of equations with a fractional Laplacian have gained also a lot of interest, and have been e.g. related to a model with modified dissipativity arising in turbulence in \cite{BPFS}. For an existence and uniqueness result in Besov spaces we refer to \cite{W}. A derivation of fractional fluid dynamical equations from kinetic transport equations would therefore be desirable to obtain. As a first step towards this direction  we  here  analyse the linear case, i.e. we start from a linear kinetic transport equation of the form 
\beq
\label{kin.equ}
 \pa_t f + v\cdot \nabla_x f = {\cal L} f\,,
\eeq
where we assume the null space of ${\cal L}$ to be spanned by the equilibrium distribution $M(v)$ satisfying
\[
M(v)=M(|v|)\geq 0, \quad M(v)<\infty, \quad \int_{\mathbb{R}^d} M(v) dv =1\,,
\]
with the moment conditions
\beq\label{mom.M}
\int_{\mathbb{R}^d} M(v) dv =1, \quad \int_{\mathbb{R}^d} |v|^2 M(v) dv = d, \quad \int_{\mathbb{R}^d} |v|^4 M(v) dv = d(d+2)\,.
\eeq
We assume in the following $M(v)$ to be either the classical Gaussian
\beq\label{M.gauss}
M^*(v)=\frac{1}{(2\pi)^{d/2}}e^{-\frac{|v|^2}{2}}\,,\eeq
or a heavy tailed distribution  satisfying
\beq\label{M.tail}
\tilde M(v)= \frac{c_0}{|v|^{\alpha+d}} \quad \textnormal{for} \ |v|\geq 1 \,
\eeq
for some $\alpha>4$, that will be specified below and for some positive constant $c_0$.
For the Gaussian $M^*(v)$ the moment conditions in \eqref{mom.M} can be easily verified. For the second class of equilibrium distributions with heavy tails we only prescribe the behaviour for $|v|\geq 1$ and assume $\tilde M(v)$ to be smooth and bounded from above and below for small velocities. Hence for $\alpha>4$ the particularly chosen constants in \eqref{mom.M} mean no loss of generality. If in the following we keep the general notation $M(v)$, the statement holds for both $M(v)=M^*(v)$ and $M(v)=\tilde M(v)$.

The macroscopic moments for density, momentum and temperature (actually, temperature times density) of $f$ are given by
\beqs
U_f=\left( \begin{array}{c}\r_f\\m_f\\\theta_f\end{array}\right)= \int_{\mathbb{R}^d} \zeta(v) f  dv, \qquad \textnormal{where} \quad \zeta(v)=\left( \begin{array}{c}1\\v\\\frac{|v|^2-d}{d}\end{array}\right)\,.
\eeqs
We consider a  linear collision operator of the form
\beq\label{def.L}
{\cal L} f= \nu(v)\left({\cal K} f - f \right)
\eeq
with the operator ${\cal K}$ being defined as
\beq\label{def.K}
{\cal K} f= M(v)\, \phi(v)\cdot U_{\nu,f}= M(v)\left(\r_{\nu,f} +  v\cdot m_{\nu,f} +\frac{|v|^2-d}{2}\theta_{\nu,f}\right)\,,
\eeq
where 
\beq\label{phi}
\phi(v)=\left( \begin{array}{c}1\\v\\\frac{|v|^2-d}{2}\end{array}\right)\,
\eeq
differs from $\zeta(v)$ only due to a normalising constant in the last component.
The collision frequency is assumed to be velocity dependent in the sense that
\[\nu(v)=\nu(|v|)\geq 0\,.\]
For the Gaussian equilibrium distribution  the corresponding collision frequency $\nu(v)=\nu^*(v)$ is assumed to have a degeneracy as $|v|\rightarrow 0$ of the following form
 \beq\label{nu.gauss}
 \nu^*(v)=|v|^{\beta^*}  \qquad \textnormal{for} \quad |v|\leq 1\,,
 \eeq
for some $\beta^*>0$ specified below. Moreover $\nu^*(v)$ is assumed to be smooth and bounded from above and below by a positive constant for $|v|\geq 1$. For the heavy-tailed equilibrium distribution the following far-field behaviour of the collision frequency $\nu(v)=\tilde \nu(v)$ is assumed
\beq\label{nu.tail}
\tilde\nu(v)=|v|^{\tilde \beta} \qquad \textnormal{for}\quad  |v|\geq 1\,,
\eeq
where $\tilde \beta<1$ will be coupled to the parameter $\alpha$ determining the tail of $\tilde M(v)$. Here $\tilde \nu(v)$ is assumed to be smooth and bounded from above and below by a positive constant for small velocities.  The macroscopic quantity  $U_{\nu,f}=(\r_{\nu,f},m_{\nu,f},\theta_{\nu,f})^T$ is defined via
\beq\label{U.nu}
\int_{\mathbb{R}^d}  \nu \phi f dv = A  U_{\nu,f}
\eeq
in such a way that the collision operator satisfies the conservation laws
\beq\label{cons.L}
\int_{\mathbb{R}^d} \phi {\cal L} f  dv =0 \,. \eeq
Using \eqref{def.K} this implies for the matrix $A$ in \eqref{U.nu}
\[A=\int_{\mathbb{R}^d} \nu \, \phi\otimes \phi M dv, 
\]
where invertibility of $A$ can be checked by direct calculation. Observe that for $f$ of the form $f=M\phi\cdot U$ we have $U_{\nu,f}=U_f=U$. We can then express the linear operator ${\cal K}$ as 
\beq
{\cal K} f = M\, \phi\cdot U_{\nu,f}
\label{K.equ}=  M\, \phi\cdot A^{-1} \int_{\mathbb{R}^d} \nu \phi  f  dv\,.
\eeq
Now the conservation properties can easily be checked 
\beqs
\int_{\mathbb{R}^d} \phi {\cal L}f \, dv &=& \int_{\mathbb{R}^d} \nu \phi f  dv - \int_{\mathbb{R}^d} \nu \phi M  \phi \cdot A^{-1}  dv \int_{\mathbb{R}^d} \nu \phi \, f \, dv\\
&=& \left(I-\int_{\mathbb{R}^d} \nu \phi \otimes \phi M  dv A^{-1}  \right) \int_{\mathbb{R}^d} \nu \phi f  dv
=\left(I-A A^{-1}  \right)\int_{\mathbb{R}^d} \nu \phi f  dv  = 0\,.
\eeqs
Clearly the vector $\phi(v)$ in \eqref{cons.L} can be replaced by the vector $\zeta(v)$, since their only difference is a normalising constant factor in the last component. Integrating  the kinetic transport equation against $\zeta(v)$, the conservation laws in terms of the macroscopic moments read: 
\beqs
&&\pa_t \r_f + \nabla\cdot m_f=0\,,\\
&&\pa_t m_f + \nabla_x \cdot \int_{\mathbb{R}^d} v \otimes v \, f \, dv =0\,,\\
&&\pa_t \theta_f + \nabla_x \cdot\int_{\mathbb{R}^d} v\frac{|v|^2-d}{d} f\, dv =0\,.
\eeqs
As mentioned above we will also investigate the limit to the fractional Stokes equation, hence in this case we shall only assume the conservation of mass and momentum. In this case we have
\beq\label{cons.mm}
\bar{\phi}(v) = \left(\begin{array}{c}1\\v\end{array}\right)\,, \qquad \bar U_f=\left(\begin{array}{c}\r_f\\m_f\end{array}\right)=\int_{\mathbb{R}^d} \bar\phi(v) f dv\,.
\eeq
Observe in particular that the corresponding $\bar A$ is a diagonal matrix.

In the remainder of introduction we are going to motivate the choice of the linear BGK model and recall the formal classical Stokes-Fourier limit as well as point out the difference to the regime with fractional rescaling. We then summarise the assumptions on the equilibrium distributions and the parameters involved and state the main results. Section 2 contains well-posedness and a priori estimates. We then introduce in a similar fashion to \cite{M} and \cite{BMP} an auxiliary function on which the moments method is based upon in Section 3 and prove the necessary convergence properties for the individual terms arising in the weak formulation. These are then unified for deriving the macroscopic dynamics in the fractional Stokes and Stokes-Fourier limit in Section 4.

Before demonstrating the classical formal Stokes-Fourier limit we shall give a brief motivation for the choice of our collision operator. Struchtrup \cite{S} and e.g. also \cite{CDL} used a power law form of $\nu$ in terms of $|v-m/\r|$ for large absolute values of the latter to obtain the correct Prandtl number out of a nonlinear BGK model of the following type:
\beq\label{nonl.BGK}
\pa_t F + v\cdot \nabla_x F =\nu(|v-m_F/\r_F|) \left({\cal M}(\r_{\nu,F},m_{\nu,F}, \theta_{\nu,F}) - F\right)\,,
\eeq
where ${\cal M}$ denotes the Maxwellian
\beqs
{\cal M}(U)=\frac{\r}{(2\pi \theta)^{d/2}}e^{-\frac{|v-m/\r|^2}{2\theta}}\,.
\eeqs
The macroscopic quantities $U_{\nu,F}$  are again defined such that the conservation laws are guaranteed:
\beqs
\int_{\mathbb{R}^d} \nu(|v-m_F/\r_F|) \,\phi(v) {\cal M}(U_{\nu,F})dv =
\int_{\mathbb{R}^d} \nu(|v-m_F/\r_F|) \,\phi(v) F \, dv\,.
\eeqs

We assume to be close to the global equilibrium ${\cal M}(1,0,1)$ (which corresponds to $M^*(v)$ from \eqref{M.gauss}). This means we can write for the remainder  $F-{\cal M}(1,0,1)=\delta  f$ for a  small parameter $\delta$. Then the linearised equation reads as follows
\beq
\pa_t f + v\cdot \nabla_x f = \nu(|v|) \left(\nabla_{U} {\cal M}(1,0,1)\cdot U_{\nu,f}  - f\right)\,,
\eeq
where $U_{\nu,f}$ is given by relation \eqref{U.nu}. Observing moreover $\nabla_U {\cal M}(1,0,1) = \phi(v) M^*(v)$ we arrive at \eqref{kin.equ} with the operator given by \eqref{def.K}.

\subsection{The (classical) Stokes-Fourier Limit}
We shall briefly outline the formal derivation of the Stokes-Fourier system as the (classical) diffusion limit from the linear kinetic transport equation with  the diffusion scaling $\gamma=2$:
\beq\label{kin.equ.class}
\e^2 \pa_t f^\e + \e v\cdot \nabla_x f^\e = \nu( M  \phi\cdot U_\nu^\e - f^\e)\,,
\eeq
where here and in the following we denote the macroscopic moments of $f^\e$ by $U^\e:=U_{f^\e}$.  For more details we refer e.g. to the work of \cite{GL}, where the limit for the Boltzmann equation is carried out. Integration in $v$ gives the macroscopic equation
\beq
\label{rho.eps}\e \pa_t \r^\e + \nabla_x\cdot m^\e=0\,,
\eeq
which is closed  in terms of the macroscopic moments. This equation formally  provides the incompressibility condition for $m$ in the limit $\e\rightarrow 0$. Integrating \eqref{kin.equ.class}
against $v$  implies
\beq \label{m.eps.class}
\pa_t m^\e + \frac{1}{\e}\nabla_x \cdot \int_{\mathbb{R}^d} v\otimes v f^\e dv =0\,.
\eeq
 We shall split the second moment as follows
\beq \label{m.eps.class.1}
\pa_t m^\e +\frac{1}{\e}\nabla_x \int_{\mathbb{R}^d} \frac{|v|^2}{d} f^\e dv+ \frac{1}{\e}\nabla_x \cdot \int_{\mathbb{R}^d} \left(v\otimes v-\frac{|v|^2}{d}I\right) f^\e dv =0\,.
\eeq
The second term can be expressed in terms of the macroscopic moments as follows:
\[\int_{\mathbb{R}^d} \frac{|v|^2}{d} f^\e dv =  \int_{\mathbb{R}^d} \left(\frac{|v|^2-d}{d}\right) f^\e dv+\int_{\mathbb{R}^d} f^\e dv=
  \theta^\e+\r^\e   \,, \]
which provides the Boussinesq relation at leading order. The remaining terms of order 1 in the equation for $m$ are of gradient type and therefore correspond to a pressure term, which vanishes when using divergence-free test functions. To analyse the behaviour of the third integral in \eqref{m.eps.class.1} we employ the macro-micro  decomposition
	\[f^\e= M \phi\cdot U_\nu^\e + g_\nu^\e\,,\]
which inserted into into the kinetic equation \eqref{kin.equ.class} formally gives
	\[g_\nu^\e = - \e \frac{v}{\nu} M \cdot \nabla_x (\phi\cdot U_\nu^\e) + O(\e^2) = - \e \frac{v}{\nu} M \cdot  ( \phi\cdot\nabla_x U^\e) + O(\e^2)\,,\]
	since knowing that $g_\nu^\e$ is $O(\e)$, implies that $U_\nu^\e=U^\e +O(\e).$ Now one can see that the macroscopic part of the antisymmetric integral term in \eqref{m.eps.class.1} vanishes and we are left with
	\beqs
	\pa_t m^\e + \frac{1}{\e} \nabla_x (\r^\e +  \theta^\e ) +\frac{1}{\e}
	\nabla_x  \cdot \int_{\mathbb{R}^d}\,
	\left(v\otimes v-\frac{|v|^2}{d}I\right) g_\nu^\e dv  =0\,.
	\eeqs
The leading order term of $g_\nu^\e$ implies
	 \beqs
	&&-\frac{1}{\e}\nabla_x  \int_{\mathbb{R}^d} \left(v\otimes v-\frac{|v|^2}{d}I\right) g_\nu^\e dv
	= \nabla_x\cdot \int_{\mathbb{R}^d} \left(v\otimes v
	-\frac{|v|^2}{d}I\right) \frac{M}{\nu} (v\otimes v \,:\, \nabla_x m^\e) dv + O(\e) \\
	&&\qquad = \mu_0 \nabla_x\cdot (\nabla_xm^\e + (\nabla_xm^\e)^T) +O(\e)  = \mu_0 \Delta_xm^\e +  O(\e) \,,
	\eeqs
	for $\mu_0=\int_{\mathbb{R}^d}v_1^2v_2^2\frac{M}{\nu}dv$, where we have used the incompressiblity condition to leading order. Summarising we obtain from the equation for $m^\e$
	\beq\label{Bouss.eps}
	&&\nabla_x (\r^\e + \theta^\e ) = O(\e)\,,\\
	\label{m.eps}
	&&\pa_t m^\e  =\mu_0 \Delta_x  m^\e+  \nabla_x  p^\e+ O(\e)\,.
	\eeq
	We shall now turn to the equation for the temperature and therefore consider the following moment
	\beq
	\pa_t \int_{\mathbb{R}^d} \frac{|v|^2-(d+2)}{2} f^\e dv + \frac{1}{\e}\nabla_x\cdot \int_{\mathbb{R}^d}  v\, \frac{|v|^2-(d+2)}{2}
 f^\e dv =0\,.
	\eeq
	Note that due to the Boussinesq equation we have 
	\[\int_{\mathbb{R}^d} \frac{|v|^2-(d+2)}{2} f^\e dv=\frac{d}{2}(\theta^\e-\r^\e)=d \theta^\e + O(\e)\,.\]
	The choice of the moment is such that inserting the decomposition into the second integral, the leading term vanishes:
	\beqs
	&&\frac{1}{\e}\nabla_x\cdot \int_{\mathbb{R}^d}  v \frac{|v|^2-(d+2)}{2} f^\e dv  \\
	&&\quad = \frac{1}{\e}\nabla_x\cdot \int_{\mathbb{R}^d}  v \otimes v \frac{|v|^2-(d+2)}{2} M m^\e dv +\frac{1}{\e}\nabla_x\cdot \int_{\mathbb{R}^d}  v \frac{|v|^2-(d+2)}{2} g^\e dv\\
	&&\quad = -\nabla_x\cdot \int_{\mathbb{R}^d}  v\otimes v \,  \frac{|v|^2-(d+2)}{2}  \frac{M}{\nu}   \nabla_x(\phi\cdot  U^\e)  dv +O(\e) \\
	&&\quad = -\nabla_x\cdot \int_{\mathbb{R}^d}  v\otimes v \,  \frac{|v|^2-(d+2)}{2}  \frac{M}{\nu}   \nabla_x\left(\r^\e+\frac{|v|^2-d}{2}\theta^\e\right)  dv +O(\e) \\
&&\quad= -d\kappa_0 \Delta_x  \theta^\e +O(\e)\,,
	\eeqs
	for $\kappa_0=\int_{\mathbb{R}^d} \frac{|v|^2(|v|^2-(d+2))^2}{4d}  \frac{M}{\nu}dv>0$, where we used the Boussinesq relation to leading order. Hence formally we arrive in the limit $\e\rightarrow 0$ at the incompressible Stokes-Fourier system:
	\beqs
	\r +  \theta=0\,,&&  \nabla_x\cdot m=0\\
	\pa_t m  &=& \mu_0 \Delta_x  m + \nabla_x p\\
	\pa_t \theta &=& \kappa_0 \Delta_x  \theta
	\eeqs
	Note that the momentum satisfies a heat equation up to a pressure gradient. This pressure term vanishes when using divergence-free testfunctions,  which are typically used for incompressible fluid dynamical equations.

\subsection{Rescaled equation for fractional Stokes-Fourier limit and function spaces}

As already mentioned in the introduction above it is our aim to analyse the Cauchy problem for the kinetic equation with a rescaling in time of order $\gamma \in (1,2)$:
 \beq
\label{resc.CP}&&\e^\gamma \pa_t f^\e +\e v\cdot \nabla_x f^\e = {\cal L}f^\e\\
&& f^\e (0,v,x)= f^{in}(v,x) \quad \in L_{x,v}^2(M^{-1})\,, \qquad \textnormal{satisfying} \quad \nabla\cdot \int_{\mathbb{R}^d}vf^{in} dv=0\,. \nonumber
\eeq
Note that the latter condition guarantees that  the initial data verifies the  incompressibility condition $\nabla_x \cdot m^{in}=0$.  Here and in the following we denote  weighted $L^2$-spaces as:
\beq\label{def.norm}
\|h\|^2_{L_{t,x,v}^2(\omega)}=\int_0^\infty\int_{\mathbb{R}^{2d}} h^2\, \omega \, dv dx dt\,.
\eeq
The weight functions we are considering in this work will only depend on $v$. To be more precise we will need  the weight functions $M^{-1}, \nu M^{-1}$ and $M$. The spaces $L^2_{x,v}(\omega)$ and $L^2_v(\omega)$ are defined in a similar way, where integration in \eqref{def.norm} is performed over $x,v$ or $v$ respectively.
Also we shall use the abbrevations $L^p_t = L^p(0,\infty)$, $L^p_{x,v}=L^p(\mathbb{R}^d\times \mathbb{R}^d)$ and $L^p_{t,x}=L^p((0,\infty)\times \mathbb{R}^d)$.

The conservation property of ${\cal L}$ implies for the zeroth moment of \eqref{resc.CP} after dividing by $\e$
\beq
\label{equ.rho.eps}\e^{\gamma-1} \pa_t \r^\e + \nabla_x \cdot m^\e=0\,,
\eeq
which provides again the incompressibility condition to leading order. 
Using the same macro-micro decomposition as above, we obtain for the first and second moment similar to before
\beqs
\pa_t m^\e + \e^{1-\gamma}\nabla_x (\r^\e+\theta^\e) &=& \e^{2-\gamma}\nabla_x \cdot \int_{\mathbb{R}^d} \left(v \otimes v -\frac{|v|^2}{d}I\right)\frac{v}{\nu} M \cdot \nabla_x(v\cdot m^\e) dv +O(\e)\,,\\
\pa_t \theta^\e&=&\e^{2-\gamma}\nabla_x\cdot \int_{\mathbb{R}^d}  \frac{ |v|^2(|v|^2-(d+2))^2}{4d}\frac{M}{\nu} dv \nabla_x\theta^\e  +O(\e)\,.
\eeqs
If we consider the fractional Stokes limit, then either the 2nd or the 6th moment of $M/\nu$ will be unbounded, but in such a way that it is balanced by the order $\e^{2-\gamma}$ in the limit $\e\rightarrow 0$. Considering the fractional Stokes limit (i.e. there is no equation for $\theta$) requires the 4th moment to be unbounded. This also explains why we cannot derive a fractional Stokes-Fourier system with a fractional Laplacian appearing in both equations for $m$ and $\theta$. 

We shall also note that the scaling $\gamma=1$ corresponds to the scaling for the acoustic limit. 

\subsection{Summary of the assumptions and results}

\begin{ass}\label{ass.par}[Assumptions on the parameters for the fractional Fourier-Stokes limit]
\begin{itemize}
\item[(i)] For the case of heavy-tailed equilibrium distributions $\tilde M$ we shall make the following assumptions on the parameters $\alpha, \tilde \beta$ determining the behaviour of $\tilde M$ and the corresponding collision frequency $\tilde \nu$ for large $|v|$ (see \eqref{M.tail} and \eqref{nu.tail}):
    
Let $\alpha >5$ and $\tilde \beta <1$ satisfy
\beq
5<\alpha +\tilde \beta < 6\,, \qquad \tilde \beta < \frac{ \alpha-4}{2}\,.
\eeq
The parameter $\tilde \gamma$ used for the rescaling in time then satisfies
\[\tilde\gamma = \frac{\alpha - \tilde\beta -4}{1-\tilde\beta}  \ \in \ (1,2)\,.\]
Observe that this also includes a velocity independent collision frequency $\tilde\nu(v)\equiv 1$. In this case the requirements on the parameters are
\beqs
\tilde\beta=0\,, \qquad \alpha = 5 + \delta \quad \textnormal{for} \quad \delta \in (0,1)\,, \qquad \tilde\gamma =1+\delta\,.
\eeqs
\item[(ii)] For the Gaussian equilibrium distributions $ M^*$ the collision frequency $\nu^*$ is degenerate as $|v|\rightarrow 0$ with exponent $\beta^*>1$, see \eqref{nu.gauss}. For this exponent $\beta^*$ and the corresponding parameter $\gamma^*$ for the rescaling in time we assume
\beqs
d+2<\beta^*<d+3\,, \qquad \gamma^* =\frac{\beta^*+d}{\beta^* -1} \ \  \in \ (1,2)\,.
\eeqs
\end{itemize}
\end{ass}
These conditions stated in Assumption \ref{ass.par} imply for the heavy-tailed equilibrium distribution the following integrability properties
\beq
\label{tail.mom}
\int_{\mathbb{R}^d} \frac{|v|^k}{\tilde \nu}\tilde M dv \leq C \ \ (k\leq 5),\qquad \int_{\mathbb{R}^d}\frac{|v|^6}{\tilde \nu}\tilde M dv=\infty \,,\eeq
whereas for the Gaussian equlibrium distribution the unboundedness occurs at the lowest order
\beq
\label{gauss.mom}
\int_{\mathbb{R}^d} \frac{|v|^2}{\nu^*} M^* dv =\infty,\qquad \int_{\mathbb{R}^d}\frac{|v|^j}{\nu^*} M^* dv\leq C \ \ (j\geq 3) \,.
\eeq
If in the following the statements do hold for both cases of equilibrium distributions in Assumption~\ref{ass.par}  we write $(M,\gamma)$, which can be either $(\tilde M,\tilde \gamma)$ or $(M^*,\gamma^*)$. 
\begin{thm}\label{thm.1}
Let Assumption \ref{ass.par} hold. Then the solution $f^{\e}$ to \eqref{resc.CP} converges as $\e\rightarrow 0$ to
\beq\label{ee.theta}
f^{\e}(t,x,v) \rightharpoonup^* f(t,x,v) = M \left(v\cdot m(x) + \frac{|v|^2-(d+2)}{2}\theta(t,x)\right)
\quad \textnormal{in} \quad L^\infty(0,T;L^2_{x,v}(\nu M^{-1}))\,,
\eeq
where the macroscopic quantities are determined by
\beqs
&&m(x) = m^{in}(x),  
 \\
&&\pa_t \theta  = -\kappa (-\Delta)^{\gamma/2} \theta, \qquad \theta (0,x)= \theta^{in}(x)\,,
\eeqs
for a positive constant $\kappa>0$, where the equations are understood in the weak sense. In particular $\pa_t m=0$ holds modulo gradients, i.e. for divergence-free testfunctions.  The initial data 
\[U^{in}=\int_{\mathbb{R}^d}\zeta(v)f^{in}(x,v) dv\]
is hereby assumed to satisfy
\[\nabla_x \cdot m^{in}(x)=0, \quad \r^{in}(x)+\theta^{in}(x)=0\,.\]
\end{thm}
The derivation of this theorem shows that one cannot obtain a fractional derivative in all moments at the same time, since the chosen time scale is not the right one for the diffusive terms in the momentum equation. For the sake of completeness we shall recall here that the fractional Laplacian can be defined using the Fourier Transform 
\beqs
{\cal F}((-\Delta_x)^{\gamma/2}h)(k)=|k|^{\gamma}{\cal F}(h)(k)\,.
\eeqs
We will rather use the following alternative representation as a singular integral
\beqs
(-\Delta_x)^{\gamma/2}h=C_{d,\gamma}PV\int_{\mathbb{R}^d}\frac{h(x)-h(y)}{|x-y|^{d+\gamma}}dy\,,
\eeqs
see e.g. also \cite{N}.
\begin{ass}
\label{ass.par.2}[Assumptions on the parameters for the fractional Stokes system without temperature]
We shall here only consider the case of heavy-tailed equilibrium distributions $\tilde M$ with corresponding collision frequency $\tilde \nu$. For the parameters $\alpha$ and $\tilde \beta$ (see \eqref{M.tail} and \eqref{nu.tail}) we  make the following assumptions:

Let $\alpha >3$ and $\tilde \beta <1$ satisfy
\beq
3<\alpha +\tilde \beta < 4\,, \qquad \tilde \beta < \frac{\alpha-2}{2}\,.
\eeq
The parameter used for the rescaling in time then satisfies
\[\tilde\gamma = \frac{\alpha - \tilde\beta -2}{1-\tilde\beta}  \ \in \ (1,2)\,.\]
Again this includes the case $\tilde \nu \equiv 1$ with the choice of parameters
\[\tilde \beta=0 \,, \qquad \alpha = 3 + \delta \quad \textnormal{for} \quad \delta \in (0,1)\,, \qquad  \tilde \gamma = 1+\delta \,.\]
\end{ass}
The corresponding conditions to \eqref{tail.mom} for these heavy-tailed equilibrium distribution read
\beq
\label{tail.mom.2}
\int_{\mathbb{R}^d} \frac{|v|^k}{\tilde \nu}\tilde M dv \leq C \ \ (k\leq 3),\qquad \int_{\mathbb{R}^d}\frac{|v|^4}{\tilde \nu}\tilde M dv=\infty\,.
\eeq

\begin{thm}\label{thm.2}
Let Assumption \ref{ass.par.2} hold. Then the solution $f^{\e}$ to \eqref{resc.CP} converges as $\e\rightarrow 0$ to
\beq\label{ee.theta}
f^{\e}(t,x,v) \rightharpoonup^* f(t,x,v) = M (\r(x) + v\cdot m(t,x))
\quad \textnormal{in} \quad L^\infty(0,T;L_{x,v}^2(\nu M^{-1}))\,,
\eeq
where the macroscopic quantities solve
\beqs
&&\r(x)\ = \ \r^{in}(x)\,,\\
&&  \nabla\cdot m=0\,, \\
&&\pa_t m = -\kappa (-\Delta)^{\tilde\gamma/2} m +\nabla_x p\,, \qquad m (0,x)= m^{in}(x)\,
\eeqs
where the equation for the evolution of $m$ holds in the weak sense. The pressure term $p\in L^2_{t,x}$ vanishes when using divergence-free testfunctions. The initial data $\bar U^{in}=\int_{\mathbb{R}^d}\bar \phi f^{in} dv$ is assumed to satisfy $\nabla\cdot m^{in}=0$.
\end{thm}
In this regime the fractional diffusion only appears in the equation for the momentum, whereas the density does not change with time. This resembles well the Navier-Stokes equations, where the density (and temperature) are assumed to be constant and the continuity equation reduces to the incompressibility condition. 
\begin{rem} The reason why the fractional Stokes limit cannot be carried out for the Gaussian equilibrium distribution is that in this case the fractional derivative arises from the unbounded second moment of $M/\nu$ and therefore appears for the density term. In the case of the Stokes-Fourier system the Boussinesq equation then relates the density to the temperature. In the Stokes limit however no such relation is available. 
\end{rem}

\section{A priori estimates and the Cauchy problem}

\subsection{Integrability conditions on $M$}
The above Assumptions \ref{ass.par} and \ref{ass.par.2}  on the parameters determining the behaviour of $M$ and $\nu$ guarantee the boundedness of the moments required for carrying out the macroscopic limit. We summarise these integrability conditions  in the following Lemma:
\begin{lem}
Let $(M,\nu)$ be either given by $(\tilde M,\tilde \nu)$ or $(M^*, \nu^*)$. In both cases we assume that the corresponding conditions on the parameters stated in Assumption \ref{ass.par} are satisfied. Then the following integrability conditions hold
\beq
\label{int.M.1} \int_{|v|\geq \delta} \frac{|v|^{2} M(v)}{\nu(v)} dv \leq C\,, \qquad \int_{\mathbb{R}^d} \frac{|v|^{j+3} M(v)}{\nu(v)} dv \leq C \qquad \textnormal{for} \  0\leq j\leq 2\,, \\
\label{int.M.2}  \int_{\mathbb{R}^d}  |v|^k \nu^2(v) M(v) dv \leq C\,, \qquad  \int_{\mathbb{R}^d}  |v|^k \nu(v) M(v) dv \leq C \qquad   \textnormal{for} \  0\leq k\leq 4\,,
\eeq
where $\delta =0$ in the case of heavy-tailed equilibrium distributions, and $0<\delta =1$ (w.l.o.g.) in the case of the Gaussian equilibrium distributions.
\end{lem}
If only the conservation of mass and momentum hold, the order of integrable moments reduces as follows:
\begin{lem}
For the heavy-tailed equilibrium distributions satisfying Assumption \ref{ass.par.2} the integrability conditions \eqref{int.M.1} hold for $j=0$ and \eqref{int.M.2} is satisfied for $0\leq k\leq 2$.
\end{lem}

\subsection{A priori estimates and well-posedness}

\begin{lem}\label{lem.cont.L}
Let the equilibrium distribution $M$ satisfy Assumption \ref{ass.par} or \ref{ass.par.2}, then
\[\|\nu {\cal K}f\|_{L_v^2(M^{-1})}\leq C \|f\|_{L_v^2(M^{-1})}\,.\]
\end{lem}
\begin{proof} The proof can be easily seen by first observing that
\beq\label{est.K.1}
\|\nu{\cal K} f\|_{L_v^2(M^{-1})} = \int_{\mathbb{R}^d} \nu^2 M (\phi\cdot U_\nu)^2 dv \leq C|U_\nu|^2\,,
\eeq
where we have used the boundedness of $M$ in \eqref{int.M.2}, which can now be employed again together with the Cauchy-Schwarz inequality to conclude
\beq\label{est.K.2}
|U_\nu|^2 = \left|A^{-1} \int_{\mathbb{R}^{d}} \nu \phi f dv\right|^2 \leq C \int_{\mathbb{R}^{d}} \frac{f^2}{M} dv \int_{\mathbb{R}^{d}} \nu^2|\phi|^2  M \, dv\leq C \|f\|^2_{L^2_v(M^{-1})}\,.\eeq
\end{proof}
This continuity property of the linear collision operator allows to deduce well-posedness of the Cauchy-problem \eqref{kin.equ} with initial data $f^{in} \in L_{x,v}^2(M^{-1})$.
The mild formulation reads
\beqs
f(t,x,v)=f^{in}(x-vt,v)e^{-\nu t}+\int_0^t e^{-\nu (t-s)}\nu{\cal K}f(s,x-(t-s)v,v)ds\,.
\eeqs
If the assumptions guaranteeing continuity of ${\cal K}$ as in Lemma \ref{lem.cont.L} hold, then a standard contraction argument yields local well-posedness, which can be extended to a global result using the  a priori estimate \eqref{bound.f} below for $\e=1$.
Clearly also the Cauchy problem for the rescaled kinetic equation is well posed for any $\e>0$:
\begin{cor}
Let Assumption \ref{ass.par} or Assumption \ref{ass.par.2} hold and let $f^{in}\in L_{x,v}^2(M^{-1})$. Then there exists a unique solution $f^\e\in L_t^{\infty}(L_{x,v}^2(M^{-1}))$ to \eqref{resc.CP}.
\end{cor}

Since we want to determine the convergence of $f^\e$ as $\e\rightarrow 0$ we shall now investigate the a priori estimates for the rescaled problem.
The basic $L^2$-estimate for kinetic transport equations is obtained by integrating the equation against $f^\e/M$. Similar to the formal derivation of the Fourier-Stokes limit in the introduction we shall introduce the micro-macro decompositions
\beq
\label{dec.g} f^{\e} &=& M\,\phi\cdot U^\e+g^\e\,,\\
\label{dec.g.nu} f^{\e} &=&  M\,\phi\cdot U_\nu^\e+g_\nu^\e \, = \, {\cal K} f^\e + g_\nu^\e\,,
\eeq
whose remainder terms fulfill
\beq
\int_{\mathbb{R}^d} \phi g^\e dv =0\,,\qquad \int_{\mathbb{R}^d} \nu \phi g^\e_\nu dv =0\,,
\eeq
due to the definition of the macroscopic moments and the conservation properties respectively. In a similar fashion to \cite{MMM} and \cite{BMP} we obtain the following lemma:
\begin{lem}\label{lem.L} Let Assumption \ref{ass.par} or Assumption \ref{ass.par.2} hold. Then the operator $\frac{1}{\nu}{\cal L}$ is bounded in $L^2_v(\nu M^{-1})$ and satisfies
\beq\label{est.L}
\int_{\mathbb{R}^d} {\cal L}f\frac{f}{M} dv = -\int_{\mathbb{R}^d}\frac{\nu}{M}|f-{\cal K} f|^2 dv
\eeq
for a positive constant $C$ and for all $f\in L^2_v(\nu M^{-1})$.
\end{lem}
\begin{proof} To prove the boundedness of $\frac{1}{\nu} {\cal L}$ it remains to check the boundedness of ${\cal K}$. In a similar fashion to \eqref{est.K.1} one can show that $\|{\cal K}f\|_{L^2_v(\nu M^{-1})}\leq C|U_\nu|^2$, and we conclude the boundedness with a slight modification of \eqref{est.K.2}:
\beqs
|U_\nu|^2=\left|A^{-1}\int_{\mathbb{R}^d}\nu \phi f dv\right|^2\leq C \int_{\mathbb{R}^d} \frac{\nu}{M}f^2 dv \int_{\mathbb{R}^d}\nu |\phi|^2 M dv \leq C\|f\|^2_{L_v(\nu M^{-1})}\,.
\eeqs
To show  \eqref{est.L} we first observe that due to the conservation properties of ${\cal L}$ \eqref{cons.L} we have
\beqs
\int_{\mathbb{R}^d}{\cal L} f \frac{{\cal K} f}{M} dv = 
\int_{\mathbb{R}^d}\phi \,{\cal L} f  \, dv \cdot U_\nu =0\,.  
\eeqs
Using this we can rewrite
\beqs
\int_{\mathbb{R}^d}{\cal L} f \frac{f}{M} dv = \int_{\mathbb{R}^d}{\cal L} f \frac{f-{\cal K}f}{M} dv = - \int_{\mathbb{R}^d}\frac{\nu}{M}|f-{\cal K} f|^2 dv \,.
\eeqs
\end{proof}
This lemma now yields the basic ingredient for deriving the following a priori estimates: 
\begin{prop}\label{prop.ex} Let Assumption \ref{ass.par} be satisfied. Then the solution $f^\e$  of \eqref{resc.CP} is bounded in $L_t^\infty( L^2_{x,v}(M^{-1}))$ uniformly with respect to $\e$. Moreover it satisfies the decomposition \eqref{dec.g.nu}, where $U_\nu^\e$ and $g_\nu^\e$ are bounded by the initial data $f^{in}$ in the sense that
\beq
\label{bound.f} \sup_{t>0} \|f^\e\|_{L^2_{x,v}(M^{-1})} &\leq& \|f^{in}\|_{L^2_{x,v}(M^{-1})}\,,\\
\label{bound.g.nu}\|g_\nu^\e\|_{L^2_{t,x,v}(\nu M^{-1})}&\leq&  \e^{\gamma/2}\|f^{in}\|_{L^2_{x,v}(M^{-1})}\,,\\
\label{bound.U.nu} \sup_{t>0} \|U_\nu^\e(t,.)\|_{L^2_{x}}& \leq&  C\|f^{in}\|_{L^2_{x,v}(M^{-1})}\,.
\eeq
\end{prop}
\begin{proof} Using \eqref{est.L}, the basic $L^2$-estimate for the solution is obtained as follows
\beqs
\frac{\e^\gamma}{2}\frac{d}{dt}\|f^\e\|^2_{L^2_{x,v}(M^{-1})}= \int_{\mathbb{R}^{2d}}{\cal L}f^\e\, \frac{f^\e}{M} dv dx = -\int_{\mathbb{R}^{2d}} \frac{\nu}{M} |f^\e-{\cal K}f^{\e}|^2 dv dx = -\int_{\mathbb{R}^{2d}} \frac{\nu}{M}(g_\nu^\e )^2 dv dx\,.
\eeqs
Integration in time implies \eqref{bound.f} and \eqref{bound.g.nu}.
For the boundedness of the macroscopic moments $U_\nu^\e$ in  \eqref{bound.U.nu}  it only remains to 
integrate \eqref{est.K.2} over $x$ and taking the supremum in time.
\end{proof}
\begin{lem}\label{lem.weak.conv} Let the assumptions of Proposition \ref{prop.ex} hold. Then
there exists a $U\in L^\infty_t(L^2_x)$, such that $f^\e \rightharpoonup^* M \phi \cdot U $ in $L^\infty((0,T);L^2_{x,v}(\nu M^{-1}))$ for any $T>0$. In particular we have the convergence of the macroscopic moments $U_\nu^\e, \, U^\e \rightharpoonup^* U$ in $L^\infty((0,T);L^2_{x})$. In the case of heavy tailed equilibrium distributions $\tilde M$ moreover strong convergence of $U_\nu^\e-U^\e\rightarrow 0$ in $L^\infty((0,T);L^2_{x})$ holds. Under Assumption \ref{ass.par.2} the same statements are valid for $\bar U^\e_\nu$ and $\bar U^\e$ respectively.
\end{lem}
\begin{proof} To see the weak$^*$-convergence we first observe that the uniform bound of $U_\nu^\e$ in $L_t^\infty(L^2_{x})$ given in \eqref{bound.U.nu} implies the existence of a  $U\in L_t^\infty(L^2_x)$ such that $U_\nu^\e \rightharpoonup^* U $ in $L_t^\infty(L^2_{x})$. Moreover the bound \eqref{bound.g.nu} implies that $f^\e-M \phi\cdot U_\nu^\e \rightarrow 0$ in $L^2_{t,x,v}(\nu M^{-1})$, which allows to deduce  $f^\e \rightharpoonup^* M \phi \cdot U $ in $L^\infty((0,T);L^2_{x,v}(\nu M^{-1}))$ for any $T>0$,  implying also for the macroscopic moment $U^\e \rightharpoonup^* U$ in $L^\infty((0,T);L^2_{x})$. 

To show the strong convergence of $U^\e_\nu-U^\e$ in the case of heavy-tailed equilibria we first note that  integrating the difference of the decompositions \eqref{dec.g}-\eqref{dec.g.nu} against $\phi$ gives
\beqs
A (U^\e_\nu-U^\e) = \int_{\mathbb{R}^{d}} \phi (g^\e - g^\e_\nu ) dv = -\int_{\mathbb{R}^{d}} \phi g^\e_\nu dv\,.
\eeqs
In the case of $M(v)=\tilde M(v)$  the integrability of $M$ in  \eqref{int.M.1} holds for $\delta =0$ and we can thus employ the Cauchy-Schwarz inequality as follows
\[\|U_\nu^\e(t,.)-U^\e(t,.)\|^2_{L^2({\mathbb{R}^{d}})}\leq C \int_{\mathbb{R}^{d}} \left(\int_{\mathbb{R}^{d}} \phi g^\e_{\nu} dv \right)^2 dx\leq 
C\int_{\mathbb{R}^{d}} \int_{\mathbb{R}^{d}} \frac{\nu (g^\e_{\nu})^2}{M} dv dx \int_{\mathbb{R}^{d}} \frac{|\phi|^2M}{\nu} dv \leq C\e^\gamma \,,\]
where for the last inequality we applied \eqref{bound.g.nu}.
\end{proof}


\section{Weak formulation and auxiliary equation}

\subsection{An auxiliary equation}
Analogously to Mellet \cite{M} and Ben-Abdallah et al. \cite{BMP} we introduce an auxiliary function $\chi^\e(t,v,x)$ defined as the solution of 
\beq\label{aux.equ}
\nu(v) \chi^\e -\e v\cdot \nabla_x \chi^\e = \nu (v) \varphi (t,x)\,,
\eeq
where $\varphi(t,x)$ is a test function in ${\cal D} ([0,\infty)\times\mathbb{R}^d)$ and hence $\chi^\e \in L_{t,v}^\infty((0,\infty)\times \mathbb{R}^d;L_x^2(\mathbb{R}^d))$. It is easy to verify that
\beqs
\chi^\e=\int_0^\infty e^{-\nu(v) z} \nu(v) \varphi(t,x+\e v z)dz\,.
\eeqs
Considering
\beq
\label{aux.rel.2}
\chi^{\e}-\varphi = \int_0^\infty \nu e^{-\nu z}(\varphi(t,x+\e vz)-\varphi(t,x))dz \,,
\eeq
it can easily be deduced that  $|\chi^\e -\varphi|\leq \|D\varphi\|_{\infty} \e|v|$, which implies uniform convergence in space and time, but not with respect to $v$. The proof of Lemma 2.5 in \cite{BMP} can easily be extended to give the following convergence results:
\beq
\label{conv.chi}\phi\, \chi^\e & \rightarrow & \phi\, \varphi \qquad \textnormal{strongly in }\ L^\infty_t(L^2_{x,v}(M))\,,\\
\label{conv.chi.t}\phi\, \pa_t\chi^\e & \rightarrow & \phi\,\pa_t\varphi \qquad \textnormal{strongly in }\ L_t^\infty(L^2_{x,v}(M))\,,
\eeq
where the extension from $\phi\equiv 1$ in \cite{BMP} to $\phi$ given as in \eqref{phi} is straightforward due to the weight $M$. The proof relies on a estimate of the form 
\beqs
\|\phi(\chi^\e-\varphi)\|^2_{L_{x,v}^2(M)} &=& \int_{\mathbb{R}^{2d}}M \left|\int_0^\infty e^{-\nu z}\nu \phi (\varphi(x+\e vz)-\varphi(x))dz\right|^2dxdv\\
&\leq& \int_{\mathbb{R}^d}\int_0^\infty M e^{-\nu z} \nu |\phi|^2 \|\varphi(\cdot+\e v z)-\varphi\|^2_{L^2_x} dz dv
\eeqs
The fact that $\|\varphi(\cdot+\e v z)-\varphi\|_{L^2_x}\rightarrow 0$ as $\e\rightarrow 0$ for all $v$ and $z$, together with the integrability condition \eqref{int.M.2}, allow to apply the Lebesgue dominated convergence theorem. A similar proof holds for the time derivative.
\subsection{The weak formulation}

Since the macroscopic equation for $\r^\e$ is closed in terms of the macroscopic moments $U^\e$ (see \eqref{equ.rho.eps}), it is sufficient to consider test functions $\varphi(t,x) \in {\cal D} ([0,\infty)\times\mathbb{R}^d)$ independent of $v$. Note that this corresponds to building the inner product in $L^2_{t,x,v}(M^{-1})$ of the kinetic equation with $\varphi(t,x) M(v)$. 
\beq
&&-\int_0^\infty \int_{\mathbb{R}^{d}} \r^{\e} \pa_t \varphi  dx dt - \int_{\mathbb{R}^{d}} \r^{in}\varphi(t=0) dx
 \label{weak.form.macro}=\e^{1-\gamma}  \int_0^\infty \int_{\mathbb{R}^{d}}\nabla_x \varphi \cdot m^{\e}  dx dt
\eeq
This equation will in the limit provide the incompressibility condition.

In order to derive equations for the macroscopic momentum and temperature we consider  the weak formulation of the rescaled kinetic equation \eqref{resc.CP} using testfunctions as introduced in the previous subsection. As for the classical Stokes-Fourier equations we shall consider the following moments corresponding to
\[
\psi(v)= \left(\begin{array}{c}v\\\frac{|v|^2-(d+2)}{2}\end{array}\right)\,.
\]
We shall for each moment $\psi_i$ consider a separated testfunction $\phi_i\in \cal D([0,\infty)\times \mathbb{R}^d)$ with its corresponding auxiliary function $\chi_i^\e$. Integrating the kinetic equation against $\psi_i \chi_i^\e$ gives
\beqs
&&-\int_0^\infty \int_{\mathbb{R}^{2d}} \psi_i f^{\e} \pa_t \chi_i^{\e} dv dx dt - \int_{\mathbb{R}^{2d}}\psi_i f^{in}\chi_i^{\e}(t=0) dv dx \\
&&\quad = \e^{-\gamma} \int_0^\infty \int_{\mathbb{R}^{2d}}\psi_i {\cal L} f^{\e}\,  \chi_i^{\e}  dv dx dt + \e^{1-\gamma} \int_0^\infty\int_{\mathbb{R}^{2d}} \psi_i v\, f^{\e}\cdot  \nabla_x\chi_i^{\e}  dv dx dt
\\
&&\quad = \e^{-\gamma}  \int_0^\infty \int_{\mathbb{R}^{2d}} \psi_i M \, \phi \cdot U^{\e}_\nu\,  \chi_i^{\e}  dv dx dt + \e^{-\gamma} \int_0^\infty \int_{\mathbb{R}^{2d}} \psi_i f^{\e} (-\nu \chi_i^{\e} + \e v\cdot \nabla_x \chi_i^{\e}) dv dx dt\\
&&\quad =\e^{-\gamma}  \int_0^\infty \int_{\mathbb{R}^{2d}} \nu\psi_i M \, \phi \cdot U^{\e}_\nu\,   \chi_i^{\e}  dv dx dt - \e^{-\gamma} \int_0^\infty \int_{\mathbb{R}^{2d}} \nu \psi_i  f^{\e} \,  dv \,\varphi_i\, dx dt\,,
\eeqs
where we have used the auxiliary equation \eqref{aux.equ}. Taking into account the conservation property of the collision operator \eqref{cons.L} in the latter integral we finally obtain the weak formulation
\beq
&&-\int_0^\infty \int_{\mathbb{R}^{2d}} \psi_i f^{\e} \pa_t \chi_i^{\e} dv dx dt - \int_{\mathbb{R}^{2d}}\psi_i f^{in}\chi^{\e}(t=0) dx dv \nonumber\\
&&\quad \label{weak.form}=\e^{-\gamma}  \int_0^\infty \int_{\mathbb{R}^{2d}}\psi_i M \, \phi \cdot U^{\e}_\nu\,   \nu(\chi_i^{\e}-\varphi_i)  dv dx dt \,.
\eeq
In the following we will analyse the convergence properties of this weak form, in particular the right hand side. In the next subsection we will analyse the limiting behaviours of the separate terms. These Lemmas will then be used in Section \ref{sec.proof} to conclude the proofs of the Theorems \ref{thm.1} and \ref{thm.2}.

\subsection{Convergence properties}
We first derive the convergence results required for  the macroscopic limit to the fractional Stokes-Fourier system. At the end of the subsection we will derive the corresponding convergence properties for the  fractional Stokes limit for conservation of density and momentum only.

In the following we will several times have to bound integrals of the form
\[
I(t,x) = \int_{\mathbb{R}^d} f(v) g(t,x + \tau v ) dv 
\]
in $L^2_{t,x}$ for some $\tau \in \mathbb{R}$. This can be done by first applying the Cauchy-Schwarz inequality and then interchanging the order of integration:
\beq
\|I\|^2_{L^2_{t,x}} &=& \int_0^\infty \int_{\mathbb{R}^{d}}\left(\int_{\mathbb{R}^d} f(v) g(t,x+\tau v) dv \right)^2 dx dt \nonumber\\
&\leq& \int_0^\infty \int_{\mathbb{R}^d} |f(v)| dv \int_{\mathbb{R}^d} |f(v)| g^2(t,x+\tau v) dv dx dt\nonumber\\
\label{est.L2}&=& \int_0^\infty \int_{\mathbb{R}^d} g^2(t,x) dx dt \left(\int_{\mathbb{R}^d} |f(v)| dv \right)^2= \|g\|_{L^2_{t,x}} \|f\|^2_{L_v^1}\,.
\eeq
We shall first consider the terms arising from the time derivative on the left hand side of the weak formulation in \eqref{weak.form}:

\begin{lem}\label{conv.time} Let Assumption \ref{ass.par} hold and let $\chi_i^\e$  be auxiliary functions satisfying \eqref{aux.equ} for $\varphi_i \in \cal D([0,\infty)\times \mathbb{R}^d)$ $(i\in\{1,\dots,d\})$. Let moreover $f_\e$ be the weak solution as in Proposition \ref{prop.ex}. Then, as $\e \rightarrow 0$, the weak form of the time derivatives in \eqref{weak.form} converges in the sense that
\beq
\int_0^\infty \int_{\mathbb{R}^{2d}}\psi_i f^{\e}\pa_t \chi_i^{\e} dv dx dt + \int_{\mathbb{R}^{2d}}\psi_i f^{in} \chi_i^{\e}(t=0) dv dx
\nonumber\\
\label{conv.time.2}\rightarrow \int_{\mathbb{R}^d} \psi_i \phi \, M\, dv \cdot \left(\int_0^\infty \int_{\mathbb{R}^{d}} U \pa_t \varphi_i dx dt + \int_{\mathbb{R}^{d}} U^{in} \varphi_i(t=0) dx\right)
\eeq
\end{lem}
\begin{proof} Due to the strong convergence of $\psi \pa_t\chi^\e_i\rightarrow \psi \pa_t \varphi_i $ in $L^\infty((0,\infty);L^2_{x,v}(M))$ in \eqref{conv.chi.t} the weak convergence of $f^\e \rightharpoonup M \phi\cdot U$ in $L^\infty((0,T);L^2_{x,v}(M^{-1}))$ and the fact that $\varphi_i$ is a test function, the stated convergence can be deduced. 
\end{proof}

For passing to the limit in the right hand side of the weak formulation in \eqref{weak.form} we will make use of the following expansions of the auxiliary function obtained by integration by parts:
\beq\label{exp.chi}
\nu(v)(\chi^{\e}(t,x,v)-\varphi(t,x)) &=& \e v \cdot \nabla_x \varphi(t,x)  + \e^2 \int_0^\infty e^{-\nu z} v^T\cdot D_x^2\varphi(t,x+\e vz)\cdot v dz\\
\label{exp.chi.2}
\nu(v)(\chi^{\e}(t,x,v)-\varphi(t,x)) &=&  \e \int_0^\infty \nu e^{-\nu z} v \cdot \nabla_x\varphi(t,x+\e vz) dz
\eeq 

We start with deriving the behaviour of the right hand side of \eqref{weak.form} for $\psi_i=v_i$ $(i\in \{1,\dots,d\})$:

\begin{lem}\label{conv.mom.1} Let the assumptions of Lemma \ref{conv.time} hold, then
\beqs
&&\e^{-\gamma} \int_0^\infty\int_{\mathbb{R}^{2d}} v_i M \, \phi\cdot U_\nu^\e\,   \nu (\chi_i^\e-\varphi_i) dv dx dt  \\
&&\quad =\e^{1-\gamma}\int_0^\infty\int_{\mathbb{R}^d}\left(\r_\nu^\e +\theta_\nu^\e\right) \pa_{x_i} \varphi_i dx dt + R^\e\qquad i\in\{1,\dots,d\}\,,
\eeqs
 where $R^\e\rightarrow 0$ as $\e\rightarrow 0$.
\end{lem}
\begin{proof} We shall employ the expansion of $\nu (\chi_i^\e-\varphi_i)$ according to \eqref{exp.chi}:
\beqs
&&\e^{-\gamma} \int_0^\infty\int_{\mathbb{R}^{2d}} v_i M \, \phi\cdot U_\nu^\e\,   \nu (\chi_i^\e-\varphi_i) dv dx dt  \\
&&\quad =\e^{1-\gamma} \int_0^\infty\int_{\mathbb{R}^{d}}\left(\int_{\mathbb{R}^d} v_i M \, \phi\cdot U_\nu^\e\, v dv \right)\cdot \nabla_x  \varphi_i \, dx dt \\
&&\qquad
+\e^{2-\gamma} \int_0^\infty\int_{\mathbb{R}^{2d}}\int_{0}^\infty v_i e^{-\nu z} v^T D_x^2 \varphi_i(t,x+\e vz) v dz M \, \phi\cdot U_\nu^\e\,  dv  dx dt \\
&&\quad =: I_1^\e + I_2^\e\,.
\eeqs
We start with showing that $I_2^\e\rightarrow 0$ performing an estimation of the type \eqref{est.L2}:
\beqs
|I_2^\e|\leq C\e^{2-\gamma}\|D_x^2\varphi_i\|_{L^2_{t,x}}\|U_\nu^\e\|_{L^2_{t,x}}\int_{\mathbb{R}^d} \frac{|v|^3+|v|^5}{\nu} M dv\leq C\e^{2-\gamma}\rightarrow 0\,.
\eeqs
The integral $I_1^\e$ gives rise to  the Boussinesq equation. The integrand of $I_1^\e$ containing the macroscopic momentum is odd and hence vanishes, such that
\beqs
I_1^\e &=& \e^{1-\gamma} \int_0^\infty\int_{\mathbb{R}^{d}} (\r_\nu^\e+ \theta_\nu^\e)  \pa_{x_i} \varphi_i  dx dt\,,
\eeqs
which concludes the proof.
\end{proof}

\begin{lem}\label{conv.frac} Let the assumptions of Lemma \ref{conv.time} hold. Then the fractional derivative arises from the following integrals as $\e \rightarrow 0$:
\begin{itemize}
\item[(i)] For the case of heavy-tailed equilibrium distributions, i.e. $M=\tilde M$ and $\nu = \tilde \nu$, we have
\[\e^{-\tilde \gamma}\int_{\mathbb{R}^{d}} \tilde \nu \tilde M \frac{|v|^4}{2d} (\chi^{\e}-\varphi)  dv  \rightarrow  -\tilde \kappa (-\Delta_x  )^{\tilde \gamma/2}\varphi
\]
strongly in $L^2_{t,x}$.
\item[(ii)] For the case of Gaussian equilibrium distributions, i.e. $M=M^*$ and $\nu = \nu^*$, we have
\[\e^{-\gamma^*} \int_{\mathbb{R}^d} \nu^* M^* (\chi^{\e}-\varphi) dv \rightarrow -\kappa^* (-\Delta_x  )^{\gamma^*/2}\varphi    \]
strongly in $L^2_{t,x}$.
\end{itemize}
\end{lem}
\begin{proof}
We shall first demonstrate the convergence for the heavy-tailed equilibrium distributions stated in (i). We therefore split the domain of integration as follows:
\beqs
\tilde J^\e_1=\e^{-\tilde \gamma}\int_{|v|\leq 1} |v|^4 \tilde \nu  \tilde M (\chi^{\e} - \varphi) dv, \quad \tilde J^\e_2=\e^{-\tilde\gamma}\int_{|v|\geq 1} |v|^4 \tilde \nu  \tilde M (\chi^{\e} - \varphi) dv\,.
\eeqs
We expand the first integral using \eqref{exp.chi}:
\beqs
\tilde J^\e_1&=&\e^{1-\tilde\gamma}\int_{|v|\leq 1} |v|^4 v   \tilde M  dv \cdot \nabla_x \varphi + \e^{2-\tilde\gamma}\int_{|v|\leq 1}\int_0^\infty |v|^4e^{-\tilde \nu z} v^TD_x^2\varphi (t,x+\e v z) v \tilde M dz dv\,.
\eeqs
The first integrand is odd, therefore the integral vanishes. The second integrand is uniformly bounded in $|v|\leq 1$, hence $\tilde J^\e_1\rightarrow 0$ as $\e\rightarrow 0$ uniformly in $t,x$ and also $L^2_{t,x}$. For the integral $\tilde J_2^\e$ we use the behaviours of $\tilde M$ and $\tilde v$, as well as \eqref{aux.rel.2}:
\beqs
\tilde J_2^\eps&=&\e^{-\tilde\gamma}c_0\int_{|v|\geq 1} |v|^{4-d-\alpha+\tilde\beta} \int_0^\infty \tilde\nu \e^{-\tilde \nu z} (\varphi(t,x+\e vz) - \varphi(t,x)) dz dv\\
&=& \e^{-\tilde\gamma} c_0 \int_{|v|\geq 1} |v|^{4-d-\alpha +\tilde \beta} \int_0^\infty e^{-s}  \left(\varphi\left(t,x+\e \frac{v}{\tilde \nu} s\right)- \varphi(t,x)\right) ds dv
\eeqs
where we  substituted $s=\tilde \nu z$. We recall that $\tilde \beta<1$ and perform the change of variables
\beq\label{change.var}
w=\e\frac{v}{|v|^{\tilde \beta}}, \qquad dv=\frac{1}{1-\tilde \beta}\left(\frac{|w|^{\tilde\beta}}{\e}\right)^\frac{d}{1-\tilde \beta} dw \,,
\eeq
where for the calculation of the determinant of the Jacobian-matrix Silvester's theorem can be applied. Recalling  $\tilde\gamma = (\alpha -\tilde\beta -4)/(1-\tilde\beta)$, we obtain
\beqs
\tilde J_2^\e&=&\frac{c_0}{1-\tilde\beta}\int_{|w|\geq \e}\int_0^\infty e^{-s}\frac{\varphi(t,x+w s)-\varphi(t,x)}{|w|^{d+\tilde\gamma }}ds dw\\
&=&\frac{c_0}{1-\tilde\beta}\int_0^\infty \int_{|y|>\e s}\frac{\varphi(t,x+y)-\varphi(t,x)}{|y|^{d+\tilde\gamma}}dy \, e^{-s} s^{\tilde\gamma} ds
\eeqs
where substituted $y= w s$.
Due to the definition of the principle value  we have the pointwise convergence in $t,x$ of
\beqs
\tilde J_2^\e \  \rightarrow \ \tilde J^0
\eeqs
with $J^0$ being defined as
\beq
\tilde J^0 &=& \frac{c_0}{1-\tilde\beta} PV \int_{\mathbb{R}^d}\int_0^\infty e^{-s}\frac{\varphi(t,x+sw)-\varphi(t,x)}{|w|^{d+\tilde\gamma}} ds dw  \nonumber\\
&=&\frac{c_0}{1-\tilde\beta}  PV \int_{\mathbb{R}^d}\frac{\varphi(t,x+y)-\varphi(t,x)}{|y|^{d+\tilde\gamma}} dy \, \int_0^\infty e^{-s} s^{\tilde\gamma} ds \nonumber\\
&=&\Gamma(1+\tilde \gamma)\, \tilde \kappa \ PV \int_{\mathbb{R}^d}\frac{\varphi(t,x+y)-\varphi(t,x)}{|y|^{d+\tilde\gamma}} dy \\
\label{J0}&=&-\tilde \kappa (-\Delta)^{\tilde \gamma/2} \varphi \nonumber
\eeq
with $\tilde \kappa = \frac{c_0\Gamma(\tilde\gamma+1)}{1-\tilde\beta}$. For proving  convergence in $L^2_{t,x}$ we proceed as in \cite{BMP} and split $\tilde J^0$ into
\beq
\frac{1}{\tilde \kappa} \, \tilde J^0&=&\int_{|w|\geq 1}\int_0^\infty e^{-s}\frac{\varphi(t,x+sw)-\varphi(t,x)}{|w|^{d+\tilde \gamma}} ds dw \nonumber \\
\label{J0.2}&&\quad + \int_{|w|\leq 1}\int_0^\infty e^{-s}\frac{\varphi(t,x+sw)-\varphi(t,x)- s w\cdot \nabla_x \varphi(t,x)}{|w|^{d+\tilde \gamma}} ds dw\,.
\eeq
These integrals are defined in the classical sense. Splitting $\tilde J_2^\e$ into the integral over the domain $\{|w|\geq 1\}$ and  $\{\e<|w|<1\}$ respectively, we obtain
\beq
\frac{1}{\tilde \kappa}(\tilde J_2^\e - \tilde J^0 )&=& - \int_{|w|\leq \e}\int_0^\infty e^{-s} \frac{\varphi(t,x+sw)-\varphi(t,x)- s w\cdot \nabla_x \varphi(t,x)}{|w|^{d+\tilde\gamma}} ds dw\nonumber\\
\label{J.conv}&=& - \int_{|w|\leq \e}\int_0^\infty e^{-s} \frac{w^T D_x^2\varphi(t,x+sw)\cdot w}{|w|^{d+\tilde\gamma}} ds dw \,,
\eeq
where we have performed integration by parts twice. Due to the fact that
\[\int_0^\infty e^{-s} ds \int_{|w|\leq \e}\frac{1}{|w|^{d+\tilde\gamma - 2}}  dw \leq C\e^{2-\tilde\gamma} \rightarrow 0\]
we deduce the (strong) $L^2_{t,x}$-convergence of $\tilde J_2^\e-\tilde J^0\rightarrow 0$, which concludes the proof for the heavy-tailed equilibrium distributions.

We shall now derive the fractional Laplacian for the Gaussian equilibrium distributions $M^*(v)=\frac{1}{(2\pi)^{d/2}}e^{-\frac{|v|^2}{2}}$ as stated in (ii). We proceed in a similar fashion to \cite{BMP} and split the integral in (ii) as follows:
\beqs
J^{\e *}_1=\e^{-\gamma^*}\int_{|v|\leq 1}  \nu^*  M^* (\chi^{\e} - \varphi) dv , \quad  J^{\e  *}_2=\e^{-\gamma^*}\int_{|v|\geq 1} \nu^*   M^* (\chi^{\e} - \varphi) dv \, .
\eeqs
 As we shall see below the degeneracy occurs in the first integral, whereas the second integral vanishes in the limit. Expanding $J^{\e *}_2$ according to \eqref{exp.chi} we obtain
\[J^{\e *}_2=\e^{1-\gamma^*}\int_{|v|\geq 1}    M^* v dv \cdot \nabla_x \varphi + \e^{2-\gamma^*}\int_{|v|\geq 1}\int_0^\infty e^{-\tilde \nu^* z} v^T D_x^2\varphi(t,x+\e v z) v M^* dz dv\, . \]
The first integral vanishes, since the integrand is odd. The second integrand is uniformly bounded  in $\{|v|\geq 1\}$, hence the second integral also converges to $0$ uniformly and in $L^2_{t,x}$.
We shall now turn to the integral $J^{\e *}_1$ over the domain of small velocities. Observe that we cannot expand $\nu^* (\chi^\e -\varphi)$ according to \eqref{exp.chi} as above, since $\int_{|v|\leq 1} \frac{|v|^2 M^*}{\nu^*} dv$ is unbounded. Hence we expand $\nu^*(\chi^\e-\varphi)$ only up to first order as given in \eqref{exp.chi.2} and proceed as in \cite{BMP}:
\beqs
J^{\e *}_1 &=& \e^{1-\gamma^*} \int_{|v|\leq 1} \int_0^\infty e^{-\nu^* z} \nu^* v\cdot \nabla_x \varphi(t,x+\e v z) dz M^* dv\\
&=& \e^{1-\gamma^*} \int_{|v|\leq 1} \int_0^\infty e^{-s}  v\cdot \nabla_x \varphi\left(t,x+\e \frac{v}{\nu^*}s\right) ds M^* dv\, .
\eeqs
We again perform a change of variables similar to \eqref{change.var}, noting that here $\beta^*>1$, such that the domain of integration is inverted:
\[w= \e \frac{v}{|v|^{\beta^*}}\,, \qquad dv=\frac{1}{\beta^*-1}\left(\frac{\e }{|w|^{\beta^* }}\right)^{\frac{d}{\beta^*-1}}dw \, .\]
Recalling $\gamma^*=(\beta^* +d)/(\beta^*-1)$ we obtain
\beqs
J^{\e *}_1 &=& \frac{1}{\beta^*-1}\int_{|w|\geq \e }\int_0^\infty e^{-s} w\cdot \nabla_x\varphi(t,x+s w) ds |w|^{-\frac{\beta^*+d}{\beta^*-1}} M^*\left((\e/|w|)^{\frac{1}{\beta^*-1}}\right)  dw \\
&=&\frac{1}{(2\pi)^{d/2}(\beta^*-1)}\int_{|w|\geq \e }\int_0^\infty e^{-s} \frac{\varphi(t,x+s w)- \varphi(t,x)}{|w|^{d+\gamma^*}} ds \,  e^{-\frac12\left(\frac{\e}{|w|}\right)^{\frac{1}{\beta^*-1}}}  dw \, .
\eeqs
As above we introduce the integral
\beqs
J^{0 *} &=& \frac{1}{(2\pi)^{d/2}(\beta^*-1)}PV \int_{\mathbb{R}^d}\int_0^\infty e^{-s}\frac{\varphi(t,x+sw)-\varphi(t,x)}{|w|^{d+\gamma^*}} ds dw\,,
\eeqs
satisfying the analogous relations given in \eqref{J.conv}. Moreover $J^{0 *}$ can be split into two integrals according to \eqref{J0.2}, from which we can deduce the $L^2_{t,x}$ convergence of $J^{\e *}\rightarrow J^{0 *}$. From the Gaussian equilibrium distributions being non-constant for small velocities two more terms arise here compared to \eqref{J.conv} and \cite{BMP}:
\beqs
&&(\beta^*-1)(2\pi)^{d/2} ( J_1^{\e *}-J^{0 *} ) = \int_{|w|\geq 1} \int_0^\infty e^{-s} \frac{\varphi(t,x+sw)-\varphi(t,x)}{|w|^{d+\gamma^*}}ds \left(e^{-\frac{1}{2}\left(\frac{\e}{|w|}\right)^{\frac{2}{\beta^*-1}}}-1\right) dw  \\
&&\qquad + \int_{\e \leq |w|\leq 1} \int_0^\infty e^{-s} \frac{\varphi(t,x+sw)-\varphi(t,x)-s w\cdot \nabla_x \varphi(t,x)}{|w|^{d+\gamma^*}}ds \left(e^{-\frac{1}{2}\left(\frac{\e}{|w|}\right)^{\frac{2}{\beta^*-1}}}-1\right) dw \\
&&\qquad  + \int_{|w|\leq \e } \int_0^\infty e^{-s} \frac{\varphi(t,x+sw)-\varphi(t,x)-s w\cdot \nabla_x \varphi(t,x)}{|w|^{d+\gamma^*}}ds dw\\
&&\quad =:L_1^{\e *} +  L^{\e *}_2 + L^{\e *}_3 \, .
\eeqs
 For the third integral $L_3^{\e *}$ the convergence to $0$ in $L^2_{t,x}$  is obtained in the same fashion to \eqref{J.conv} above. For $L_1^{\e *}$ we employ an estimation as in \eqref{est.L2}:
\beqs
\|L_1^{\e *}\|_{L^2_{t,x}}\leq 2 \|\varphi\|_{L^2_{t,x}}\int_0^s e^{-s} ds \int_{|w|\geq 1}|w|^{-(d+\gamma^*)} \left(1-e^{-\frac12\left(\frac{\e}{|w|}\right)^{\frac{2}{\beta^*-1}}}\right)dw \leq C \left(1-e^{-\frac12\e^{\frac{2}{\beta^*-1}}}\right)\rightarrow 0
\eeqs
To see the convergence of the remaining term $L_2^{\e *}$ we perform integration by parts twice and  bound
\beqs
\|L_2^{\e *}\|_{L^2_{t,x}}\leq \|D_x^2\varphi\|_{L^2_{t,x}}\int_{\e \leq |w|\leq 1}\int_0^s e^{-s} |w|^{-(d+\gamma^*-2)} \left(1-e^{-\frac12\left(\frac{\e}{|w|}\right)^{\frac{2}{\beta^*-1}}}\right)ds dw
\eeqs
We now split the domain of integration  in the latter integral once more. For any $a \in(0,1)$
\beqs
&&\int_{\e \leq |w|\leq 1} \int_0^\infty \frac{e^{-s}}{|w|^{d+\gamma^*-2}}\left(1- e^{-\frac12\left(\frac{\e}{|w|}\right)^{\frac{2}{\beta^*-1}}}\right) ds dw\leq C\int_{\e }^1
r^{1-\gamma^*}\left(1- e^{-\frac{1}{2}\left(\frac{\e}{r}\right)^{\frac{2}{\beta^*-1}}} \right) dr\\
&&\quad = C \int_{\e c^*}^{\e^a}
r^{1-\gamma^*}\left(1- e^{-\frac{1}{2}\left(\frac{\e}{r}\right)^{\frac{2}{\beta^*-1}}} \right)  dr +
\int_{\e^a}^{1}
r^{1-\gamma^*}\left(1- e^{-\frac{1}{2}\left(\frac{\e}{r}\right)^{\frac{2}{\beta^*-1}}} \right) dr \\
&&\quad \leq C r^{2-\gamma^*}\big|_{\e }^{\e^{a}} +
C \left(1- e^{-\frac{
\e^{\frac{2(1-a)}{\beta^*-1}}}{2}
}\right) \rightarrow 0\,.
\eeqs
By dominated convergence, this implies the strong convergence of $J_2^{\e*}$ to $J^{0*}$ in $L^2_{t,x}$, which concludes the proof of the Lemma.
\end{proof}

\begin{lem}\label{conv.mom.2} Let the assumptions of Lemma \ref{conv.time} hold and recall that $\psi_{d+1}=\frac{|v|^2-(d+2)}{2}$. Then, as $\e \rightarrow 0$, we have
\beqs
&&\e^{-\gamma} \int_0^\infty\int_{\mathbb{R}^{2d}} \psi_{d+1} M \, \phi\cdot U_\nu^\e \,  \nu (\chi^\e-\varphi) \, dv dx dt  \  \rightarrow \   -\kappa \int_0^\infty\int_{\mathbb{R}^d}\theta (-\Delta)^{\gamma/2} \varphi \, dx dt \,.
\eeqs
\end{lem}
\begin{proof} We shall again employ the expansion of $\nu (\chi^\e-\varphi)$ according to \eqref{exp.chi}:
\beqs
&&\e^{-\gamma} \int_0^\infty\int_{\mathbb{R}^{2d}} \psi_{d+1} M \, \phi\cdot U_\nu^\e \,  \nu (\chi^\e-\varphi) dv dx dt  \\
&&\quad =\e^{1-\gamma} \int_0^\infty\int_{\mathbb{R}^{d}}\int_{\mathbb{R}^d} \psi_{d+1} M \, \phi \cdot U_\nu^\e v\,  dv \cdot \nabla_x \varphi\,  dx dt \\
&&\qquad
+\e^{2-\gamma} \int_0^\infty\int_{\mathbb{R}^{2d}}\int_{0}^\infty \psi_{d+1} e^{-\nu z} v^T D_x^2 \varphi(x+\e vz,t) v\,  dz M \, \phi \cdot U_\nu^\e\,  dv dx dt \\
&&\quad =: I_1^\e + I_2^\e \, .
\eeqs
The part in the integrand of $I_1^\e$ containing the macroscopic density and temperature is odd and hence vanishes, therefore we are left with computing only the part containing the momentum:
\beqs
2 I_1^\e &=&\e^{1-\gamma} \int_0^\infty\int_{\mathbb{R}^{d}} \left(\int_{\mathbb{R}^d}\left(|v|^2-(d+2)\right) v\otimes v M \, dv\cdot m_\nu^\e \right) \cdot \nabla_x \varphi \, dx dt \ =\ 0\, ,
\eeqs
which holds due to the moment conditions in \eqref{mom.M}.
 We now turn to the second integral term $I_2^\e$, which gives rise to the fractional Laplacian.  We first order the moments accordingly
\beqs
&&2\, I_2^\e = \e^{2-\gamma}\int_0^\infty\int_{\mathbb{R}^{2d}}\int_{0}^\infty (|v|^2-(d+2)) e^{-\nu z} v^T D_x^2 \varphi(t,x+\e vz)\, v dz M \, \phi \cdot U_\nu^\e \, dv  dx dt \\
&&\quad = \e^{2-\gamma} (d+2) \int_0^\infty\int_{\mathbb{R}^{2d}}\int_{0}^\infty e^{-\nu z} v^T D_x^2 \varphi(x+\e vz,t) v dz M  dv \left(-\r_\nu^\e + \frac{d}{2}\theta_\nu^\e  \right) dx dt \\
&&\qquad + \e^{2-\gamma} \int_0^\infty\int_{\mathbb{R}^{2d}}\int_{0}^\infty \left((|v|^2-(d+2)) v\cdot m_\nu^\e+|v|^2 \left( \r_\nu^\e- (d+1) \theta_\nu^\e \right) \right)\cdot\\
&&\qquad \quad \cdot e^{-\nu z} v^T D_x^2 \varphi(x+\e vz,t) v dz M dv  dx dt \\
&&\qquad + \frac{\e^{2-\gamma}}{2} \int_0^\infty\int_{\mathbb{R}^{2d}}\int_{0}^\infty |v|^4\theta_\nu^\e e^{-\nu z} v^T D_x^2 \varphi(x+\e vz,t) v dz M  dv  dx dt \\
&&\quad =: L_1^\e+ L_2^\e+L_3^\e \, .
\eeqs
 We start with showing that $L_2^\e\rightarrow 0$ for both cases of equilibrium distributions due to \eqref{tail.mom} and \eqref{gauss.mom}
\beqs
|L_2^\e|\leq C \e^{2-\gamma}\|D_x^2\varphi\|_{L^2_{t,x}}\|U_\nu^\e\|_{L^2_{t,x}}\int \frac{|v|^3+|v|^5}{\nu} M dv \leq C\e^{2-\gamma}\rightarrow 0\,.
\eeqs
Moreover for the heavy-tailed equilibrium distributions  the integral term $L_1^\e$ also vanishes in the limit due to \eqref{tail.mom} using the same argumentation. The third integral term $L_3^\e$ corresponds, after integration by parts twice and inserting the definition of $\nu (\chi^\e -\varphi)$, to the integral in Lemma  \ref{conv.frac} (i) and hence converges towards the fractional Laplacian. For the case of Gaussian equilibrium the roles of the integrals $L_1^\e$ and $L_3^\e$ are interchanged, namely $L_3^\e$ vanishes and from $L_1^\e$ we obtain the fractional Laplacian according to Lemma \ref{conv.frac} (ii).
\end{proof}

We shall now state the corresponding convergence properties for the fractional Stokes limit without temperature. In fact, in the weak form \eqref{weak.form} we only need to consider the moment $\bar\psi (v)=v$. Since in this case we only treat the case of heavy-tailed equilibrium distributions as stated in Assumption \ref{ass.par.2}, no distinction between the types of equilibrium distributions has to made here. Hence for the fractional Stokes limit we skip the tildes for $M$ and $\nu$ in the following. 

\begin{lem}\label{conv.stokes} Let Assumption \ref{ass.par.2} hold and let $\chi_i^\e$ be the auxiliary functions as defined above \eqref{aux.equ} for corresponding $\varphi_i \in {\cal D}((0,\infty)\times \mathbb{R}^d)$ and let $f^\e$ be the weak solution as in Proposition \ref{prop.ex}.
\begin{itemize}
\item[(i)] The weak form of the time derivatives in \eqref{weak.form} for $\bar \psi=v$ converges in the sense of Lemma \ref{conv.time} with the macroscopic moments $U$ being replaced by $\bar U$ as $\e \rightarrow 0$.
\item[(ii)] For  $\bar \psi_i =v_i$ we have for the right hand side in the weak formulation of \eqref{weak.form}:
\beqs
&&\e^{-\gamma} \int_0^\infty\int_{\mathbb{R}^{2d}} v_i M \, \bar \phi\cdot \bar U_\nu^\e \,  \nu (\chi_i^\e-\varphi_i) \, dv dx dt \\
 &&\quad =  -\e^{1-\gamma} \int_0^\infty\int_{\mathbb{R}^{d}} \varphi_i \pa_{x_i}\r^\e_\nu dx dt-\kappa \int_0^\infty\int_{\mathbb{R}^d} m_i (-\Delta)^{\gamma/2} \varphi_i \, dx dt + \bar R^\e_i \,,
\eeqs
where $\bar R_i^\e\rightarrow 0$ for all $i\in \{1,\dots, d\}$.
\end{itemize}
\end{lem}
\begin{proof}
The convergence of the terms involving time derivatives in (i) is similar to the proof of Lemma \ref{conv.time}. To derive the integral identity in (ii) we first split the integral into the terms containing $\r_\nu^\e$ and $m^\e_\nu$ respectively:
\beqs
&&\e^{-\gamma} \int_0^\infty\int_{\mathbb{R}^{2d}} v_i  M \, \bar \phi\cdot \bar U_\nu^\e \,  \nu (\chi_i^\e-\varphi_i) \, dv dx dt \\
 &&\quad = \e^{-\gamma} \int_0^\infty\int_{\mathbb{R}^{2d}} v_i  M \r_\nu^\e \,  \nu (\chi_i^\e-\varphi_i) \, dv dx dt   + \e^{-\gamma} \int_0^\infty\int_{\mathbb{R}^{2d}} v_i  M \,v\cdot m_\nu^\e   \nu (\chi_i^\e-\varphi_i) \, dv dx dt\\
 &&\quad =: \bar I_1^\e + \bar I_2^\e \, .
\eeqs
We expand $ \nu(\chi_i^\e-\varphi_i)$ according to \eqref{exp.chi} in $\bar I_1^\e$:
\beqs
\bar I_1^\e&=&\e^{1-\gamma}\int_0^\infty\int_{\mathbb{R}^{2d}}v_i v  M dv \cdot \nabla_x \varphi_i \, \r_\nu^\e dx dt+ \e^{2-\gamma}  \int_0^\infty\int_{\mathbb{R}^{2d}}\int_0^\infty e^{- \nu z} v_i v D_x^2\varphi(t,x+\e v z) v   M\, dz dv\r_\nu^\e dx dt\\
&=& \e^{1-\gamma} \int_0^\infty\int_{\mathbb{R}^{d}}\r_\nu^\e \pa_{x_i}\varphi_i dx dt + \hat R_i^\e
\eeqs
where the latter integral vanishes in the limit $\e\rightarrow 0$:
\[
 |\hat R_i^\e |\leq C\e^{2-\gamma}\|D^2_x\varphi_i\|_{L^2_{t,x}}\|\r_\nu^\e\|_{L^2_{t,x}}\int_{\mathbb{R}^d}\frac{|v|^3}{\nu}M\, dv\leq C\e^{2-\gamma} \rightarrow 0 \,.
\]
We shall now derive the fractional Laplacian from the integral  $\bar I_2^\e$ and therefore, similar to above, split the integral into
\beqs
\bar I_2^\e &=& \e^{-\gamma} \int_0^\infty\int_{\mathbb{R}^{d}}\int_{|v|\leq 1} v_i M \,v\cdot m_\nu^\e   \nu (\chi_i^\e-\varphi_i) \, dv dx dt \\
&&\quad +
 \e^{-\gamma} \int_0^\infty\int_{\mathbb{R}^{d}}\int_{|v|\geq 1} v_i  M \r_\nu^\e \,v\cdot m_\nu^\e \nu (\chi_i^\e-\varphi_i) \, dv dx dt\\
 &=:& \bar J^\e_1 + \bar J^\e_2 \,.
\eeqs
Inserting \eqref{exp.chi} it is easy to see that $\bar J^\e_1$ vanishes in the limit $\e\rightarrow 0$.
We insert \eqref{aux.rel.2} in the integrand of $\bar J^\e_2$  to obtain after substituting $s=\nu z$
\beqs
\bar J_2^\e&=&\e^{- \gamma} \int_0^\infty\int_{\mathbb{R}^{d}}\int_{|v|\geq 1}\int_0^\infty \nu e^{-s}v_i v\cdot m_\nu^\e  M \left(\varphi_i\left(t,x+\e \frac{v}{\nu} s\right)-\varphi_i(t,x)\right) ds dv dx dt
\eeqs
Recalling the definition of $ \gamma = (\alpha- \beta-2)/(1-\beta)$ and using the same change of variables as in \eqref{change.var} we obtain
\beqs
&&(1- \beta)(\gamma +d)\bar I_2^\e\\
&&\quad =( \gamma +d)\int_0^\infty\int_{\mathbb{R}^{d}}\left(\int_{|w|\geq \e}\int_0^\infty e^{-s}\frac{w_i w}{|w|^2}\frac{1}{|w|^{\gamma+d}}(\varphi_i(t,x+sw)-\varphi_i(t,x))ds dw  \right)\cdot m_\nu^\e dx dt\\
&&\quad=-\int_0^\infty\int_{\mathbb{R}^{d}}\left(\int_{|w|\geq \e}\int_0^\infty e^{-s}\nabla_w\left(\frac{1}{|w|^{\gamma+d}}\right)w_i(\varphi_i(t,x+sw)-\varphi_i(t,x))ds dw  \right)\cdot m_\nu^\e dx dt\\
&&\quad=\int_0^\infty\int_{\mathbb{R}^{d}}\left(\int_{|w|\geq \e}\int_0^\infty e^{-s}\frac{1}{|w|^{\gamma+d}}(\varphi_i(t,x+sw)-\varphi_i(t,x))ds dw  \right) m_{\nu i}^\e dx dt\\
&&\qquad + \int_0^\infty\int_{\mathbb{R}^{d}}\left(\int_{|w|\geq \e}\int_0^\infty e^{-s}\frac{w_i}{|w|^{\gamma+d}}s \nabla_x \varphi_i(t,x+sw)ds dw  \right) \cdot  m_{\nu }^\e dx dt  \\
&&\qquad + \int_0^\infty\int_{\mathbb{R}^{d}}\left(\int_{|w|= \e}\int_0^\infty e^{-s}\frac{w_i}{|w|^{\gamma+d}}s  \varphi_i(t,x+sw) \frac{w}{|w|} ds d\sigma   \right)\cdot m_{\nu }^\e dx dt  \\
&&\quad=: \bar L^\e_1 +\bar L^\e_2+\bar b^\e\,,
\eeqs
where we performed integration by parts and used the fact that the outer unit normal on the sphere  is $w/|w|$. The convergence of $\bar L^\e_1$ towards the integral involving the fractional Laplacian
\beqs
\bar L_1^\e \rightarrow \kappa \int_0^\infty\int_{\mathbb{R}^d}m_i(-\Delta)^{\frac{ \gamma}{2}} \varphi_i \, dx dt
\eeqs
is deduced as in the proof of Lemma \ref{conv.mom.2}. Hence to conclude the proof it remains to show that $\bar L^\e_2$ and $\bar b^\e$ vanish in the limit. Therefore we first observe
\beqs
(1- \beta)(\gamma +d)\bar L^\e_2&=&\int_0^\infty\int_{\mathbb{R}^{d}}\left(\int_{|w|\geq \e}\int_0^\infty e^{-s}\frac{w_i}{|w|^{\gamma+d}}s \nabla_x \varphi_i(t,x+sw)ds dw  \right)\cdot m_{\nu}^\e dx dt\\
&=&-\int_0^\infty\int_{\mathbb{R}^{d}}\int_{|w|\geq \e}\int_0^\infty e^{-s}\frac{w_i}{|w|^{\gamma+d}}s ( \nabla_x \cdot m_{\nu }^\e ) \varphi_i(t,x+sw)ds dw  dx dt \\
&=& - \int_0^\infty\int_{\mathbb{R}^{2d}}\int_0^\infty e^{-s}\frac{w_i}{|w|^{\gamma+d}}s (\nabla_x \cdot m_{\nu }^\e ) \varphi_i(t,x+sw)ds dw  dx dt\\
&& + \int_0^\infty\int_{\mathbb{R}^{d}}\int_{|w|\leq \e}\int_0^\infty e^{-s}\frac{w_i}{|w|^{\gamma+d}}s (\nabla_x \cdot m_{\nu }^\e) \varphi_i(t,x+sw)ds dw  dx dt\\
&=:&\bar K_1^\e+\bar K_2^\e \,.
\eeqs
For the first integral $\bar K_1^\e$ we shall use the fact that $\nabla\cdot m_\nu^\e\rightharpoonup 0$ in $L^2_{t,x}$. Hence, if
\beq\label{int.K1}
\int_{\mathbb{R}^{d}}\frac{w_i}{|w|^{\gamma+d}}\int_0^\infty se^{-s} \varphi_i(t,x+sw)dsdw \eeq
is bounded in $L^2_{t,x}$, then $\bar K_1^\e \rightarrow 0$. Proceeding as in \eqref{est.L2} we can bound the $L_{t,x}^2$-norm of the integral \eqref{int.K1} over the domain $\{|w|\geq1\}$ directly by
\beqs
 C \|\varphi_i\|_{L^2_{t,x}}\int_1^\infty |w|^{- \gamma -d+1}dw\leq C\,.
\eeqs
For the integral \eqref{int.K1} over the domain $\{|w|\leq 1\}$ we observe that $se^{-s}=\pa_s((s+1)e^{-s})$. Integrating by parts in $s$ we can then bound the $L^2_{t,x}$-norm using an estimation of the type \eqref{est.L2} by
\beqs
 C \|\nabla_x\varphi\|_{L^2_{t,x}}\int_{|w|\leq 1}|w|^{-\gamma-d+2}dw\leq C
\eeqs
from which we can now deduce $\bar K_1^\e\rightarrow 0$ (note that the boundary term is odd in $w$ and hence vanishes).
To see $\bar K_2^\e\rightarrow 0$ we integrate by parts additionally in $x$
\[
|K_2^\e|\leq C\|m_\nu^\e\|_{L^2_{t,x}}\|D_x^2\varphi\|_{L^2_{t,x}}\int_{|w|\leq \e} |w|^{-\gamma -d+2} dw \leq C\e^{2-\gamma}\rightarrow 0\,.
\]
It now remains to show that the boundary terms vanish. We employ integration in parts twice
\beqs
|\bar b^\e|&=&\left|\frac{1}{ \gamma +d}\int_0^\infty \int_{\mathbb{R}^d}\int_{|w|=\e}\int_0^\infty e^{-s} w D_x^2\varphi(t,x+ sw) w \frac{w_i}{|w|^{ \gamma +d}}m_\nu^\e\cdot \frac{w}{|w|} d\sigma dw dx dt \right|\\
&\leq& C\|D_x^2\varphi\|_{L^2_{t,x}}\|m_\nu^\e\|_{L^2_{t,x}}\int_{|w|=\e} \frac{|w|^4}{|w|^{ \gamma +d +1}} d\sigma \leq C \e^{2- \gamma} \rightarrow 0\,.
\eeqs
\end{proof}

\section{Derivation of the macroscopic dynamics}\label{sec.proof}
\subsection{Derivation of the fractional Stokes-Fourier system}
The convergence of the solution $f^{\e}$ of the Cauchy problem in \eqref{resc.CP} was already shown. We will now derive the macroscopic equations determining the limiting solution stated in Theorem \ref{thm.1}.

\begin{proof}[Proof of Theorem \ref{thm.1}]
We start by deriving the incompressibility condition from equation \eqref{weak.form.macro}. Since $\pa_t \varphi$ and $\r^\e$ are both uniformly bounded in $L^2_{t,x}$,  multiplying \eqref{weak.form.macro} with $\e^{\gamma-1}$ and using the fact that $m^\e\rightharpoonup m$ in $L^2_{t,x}$ we obtain the incompressibility condition in the limit $\e\rightarrow 0$.

We shall now turn to the weak form of the first moments. Due to Lemma \ref{conv.mom.1} we know that
\beq
&&-\int_0^\infty \int_{\mathbb{R}^{2d}} v_i f^{\e} \pa_t \chi_i^{\e} dv dx dt - \int_{\mathbb{R}^{2d}}v_i f^{in}\chi_i^{\e}(t=0) dx dv \nonumber\\
&&=\frac{\e^{1-\gamma}}{d}\int_0^\infty\int_{\mathbb{R}^d}\left(\r_\nu^\e +\theta_\nu^\e\right) \pa_{x_i} \varphi_i dx dt +  R_i^\e\,,\qquad i\in\{1,\dots,d\} \, .
\eeq
Again, due to the boundedness of the terms on the left hand side and the remainder $R^\e$, which vanishes in the limit $\e\rightarrow 0$, we obtain after multiplying by $\e^{\gamma-1}$:
 \beq\label{Bous.eps}
 \left|\int_0^\infty\int_{\mathbb{R}^d}\left(\r_\nu^\e +\theta_\nu^\e\right) \pa_{x_i} \varphi_i dx dt \right|\leq C\e^{\gamma-1} \qquad \textnormal{for all} \ i\in\{1,\dots,d\}\,.
 \eeq
 Hence using the fact that $U_\nu^\e\rightharpoonup U$ in $L^2_{t,x}$, we obtain
 the Boussinesq relation. Moreover, carrying out the limit in the equation for $m^\e$ we obtain
\beqs
\pa_t m = \nabla_x p
\eeqs
in the weak sense, where $p(t,x)$ is the remainder of the Boussinesq relation:
 \beqs
 p(t,x)=\lim_{\e\rightarrow 0}\e^{1-\gamma}\left(\r_\nu^\e +\theta_\nu^\e\right)   = \lim_{\e\rightarrow 0}\e^{1-\gamma}\left(\r_\nu^\e - \r +(\theta_\nu^\e-\theta)\sqrt{2/d}\right)
 \eeqs
 which is bounded in $L^2_{t,x}$ due to \eqref{Bous.eps}. Using divergence-free testfunctions, i.e. $\sum_{i}\partial_{xi}\varphi_i=0$, we obtain $\pa_t m=0$.

 We shall now turn to the equation for $\theta$. Herefore we use the weak form of the moment corresponding to $\psi_{d+1}=\frac{|v|^2-(d+2)}{2}$. Lemma \ref{conv.time} and the Boussinesq relation imply
 \beqs
&&-\int_0^\infty \int_{\mathbb{R}^{2d}} \psi_{d+1} f^{\e} \pa_t \chi^{\e} dv dx dt - \int_{\mathbb{R}^{2d}}\psi_{d+1} f^{in}\chi^{\e}(t=0) dx dv\\
&&\quad \rightarrow \ \left(1+\frac{d}2\right)\int_0^\infty \int_{\mathbb{R}^{2d}} \theta\pa_t \varphi  dx dt - \left(1+\frac{d}2\right)\int_{\mathbb{R}^{2d}}\theta^{in} \varphi (t=0) dx dv
 \eeqs
 where we have used the Boussinesq equation for the limiting solution and the assumption on the initial data $\r^{in} + \theta^{in}= 0.$
 Lemma \ref{conv.mom.2} completes the derivation of the dynamics for the limiting function $f= M\phi\cdot U$.
 \end{proof}

\subsection{Derivation of the dynamics for fractional Stokes limit}
We finally for the limiting solution stated in Theorem \ref{thm.2}.
\begin{proof}[Proof of Theorem \ref{thm.2}]
The incompressibility condition from equation \eqref{weak.form.macro} can be deduced as in the proof of Theorem \ref{thm.1} above. Lemma \eqref{conv.stokes} implies
\beqs
 &&-\int_0^\infty\int_{\mathbb{R}^{2d}} v_i  f^\e \pa_t\chi_i^\e dv dx dt -\int_{\mathbb{R}^{2d}} v_i f^{in} \chi_i^\e(t=0) dv dx  \\
 &&\quad =  -\e^{1-\gamma} \int_0^\infty\int_{\mathbb{R}^{2d}} \varphi_i \pa_{x_i}\r^\e_\nu dx dt-\kappa \int_0^\infty\int_{\mathbb{R}^d} m_i (-\Delta)^{\gamma/2} \varphi_i \, dx dt + \bar R^\e_i \,,
\eeqs
where $\bar R_i^\e\rightarrow 0$ as $\e\rightarrow 0$. Using divergence-free  testfunctions, i.e. considering $\Phi=(\varphi_1,\dots, \varphi_d)^T$ with $\nabla\cdot \Phi =0$, we obtain in the limit
\beqs
&&-\int_0^\infty\int_{\mathbb{R}^{d}} m  \cdot \pa_t \Phi  dx dt -\int_{\mathbb{R}^{d}} m^{in}\cdot \Phi(t=0) dx  = \kappa \int_0^\infty\int_{\mathbb{R}^d} m\cdot (-\Delta)^{\gamma/2} \Phi\, dx dt \,,
\eeqs
which gives 
\beqs
&&\pa_t m = -\kappa (-\Delta_x  )^{-\frac{\gamma}{2}} m \\
&&m(0,x)=m^{in}(x)
\eeqs
in the distribution sense for divergence-free testfunctions.
\end{proof}

\end{document}